\documentclass[letterpaper, 10 pt, conference]{ieeeconf}  

\IEEEoverridecommandlockouts                              

\usepackage[bottom]{footmisc}

\usepackage{amsthm}
\usepackage{amssymb}
\usepackage{mathrsfs}
\usepackage{amsmath}
\usepackage{mathtools}
\usepackage{enumerate}
\usepackage[dvipsnames]{xcolor}
\usepackage[noadjust]{cite}
\usepackage{indentfirst}
\usepackage{graphicx}
\graphicspath{{./epsfiles/}}
\usepackage[font=footnotesize]{caption}
\usepackage[font=footnotesize]{subcaption}
\usepackage[breaklinks=true, colorlinks, bookmarks=true]{hyperref}
\usepackage{tikz}
\usepackage{adjustbox}
\usepackage{pgfplots}
\pgfplotsset{compat=1.9}

\newtheorem{theorem}{Theorem}[section]
\newtheorem{lemma}[theorem]{Lemma}

\newtheorem{corollary}[theorem]{Corollary}
\newtheorem{proposition}[theorem]{Proposition}
\newtheorem{problem}{Problem}
\newtheorem{remark}[theorem]{Remark}
\newtheorem{assumption}{Assumption}
\newtheorem{condition}{Condition}

\newcommand{\oprocendsymbol}{\hbox{$\bullet$}}
\newcommand{\oprocend}{\relax\ifmmode\else\unskip\hfill\fi\oprocendsymbol}

\newcommand{\longthmtitle}[1]{\mbox{}\textup{\emph{(#1):}}}

\newcommand{\stkout}[1]{\ifmmode\text{\sout{\ensuremath{#1}}}\else\sout{#1}\fi}

\definecolor{paleGreen}{rgb}{.3, .7, .3}

\newcommand{\co}[1]{\operatorname{co}(#1)}

\newcommand{\until}[1]{\{1,\dots,#1\}}
\DeclareMathOperator*{\argmax}{arg\,max}
\DeclareMathOperator*{\argmin}{arg\,min}

\newcommand{\interior}[1]{\operatorname{int}(#1)}
\newcommand{\proj}[2]{\operatorname{Proj}_{#1}(#2)}

\newcommand{\real}{\mathbb{R}}

\newcommand{\Cc}{\mathcal{C}}
\newcommand{\Xc}{\mathcal{X}}

\newcommand{\Sc}{\mathcal{S}}
\newcommand{\Uc}{\mathcal{U}}

\newcommand{\Ic}{\mathcal{I}}

\newcommand{\Lc}{\mathcal{L}}

\renewcommand{\theenumi}{(\roman{enumi})}

\allowdisplaybreaks
\renewcommand{\baselinestretch}{0.985}

\begin{document}
\title{\LARGE \bf Safe Control of Second-Order Systems with
  Linear~Constraints}

\author{Mohammed Alyaseen \quad Nikolay Atanasov \quad Jorge Cort\'es
  \thanks{M. Alyaseen, N. Atanasov, and J. Cort\'es are with the
    Contextual Robotics Institute, UC San Diego,
    \texttt{\{malyasee,natanasov,cortes\}@ucsd.edu}. M. Alyaseen is
    also affiliated with Kuwait University as a holder of a
    scholarship.} %
}

\maketitle

\begin{abstract}
  Control barrier functions (CBFs) offer a powerful tool for enforcing
  safety specifications in control synthesis. This paper deals with
  the problem of constructing valid CBFs. Given a second-order system
  and any desired safety set with linear boundaries in the position
  space, we construct a provably control-invariant subset of this
  desired safety set. The constructed subset does not sacrifice any
  positions allowed by the desired safety set, which can be nonconvex.
  We show how our construction can also meet safety specification on
  the velocity. We then demonstrate that if the system satisfies
  standard Euler-Lagrange systems properties then our construction can
  also handle constraints on the allowable control inputs. We finally
  show the efficacy of the proposed method in a numerical example of
  keeping a 2D robot arm safe from collision.
\end{abstract}


\section{Introduction}
Control barrier functions (CBF) provide a flexible framework in
safety-critical control to certify the forward invariance of a desired
set with respect to system trajectories and design feedback
controllers that ensure it.  Because of their versatility, CBFs have
made their way into numerous applications in robotics, transportation,
power systems, and beyond.  By definition, every boundary point of the
CBF's 0-superlevel set admits a control value that holds the system's
trajectory from instantaneously leaving it. This point-wise condition
is known as the CBF condition.
Finding CBFs is a challeging task:
it amounts to finding a set whose states can be made safe, i.e., for
which control actions ensuring safety can be identified. This is not
trivial given the complexity of the dynamics and limitations on the
control inputs. After clearing this challenge, one must still figure
out whether a well-behaved control law can be synthesized out of all the
point-wise safe control actions.
%
%
%
In this work, we construct valid CBFs that enforce any positional
safety requirements with linear boundaries for second-order dynamics
and provide an associated continuous safe controller.
%
%

\emph{Literature review:} Whether from Nagumo's theorem \cite{MN:42}
or from comparison results in the theory of differential equations
\cite{HK:02}, the CBF condition was first derived for smooth functions
\cite{PW-FA:07,ADA-SC-ME-GN-KS-PT:19,RK-ADA-SC:21}. This condition was
extended to non-smooth functions in multiple works, such as
\cite{MG-AI-RGS-WED:22f,PG-JC-ME:17-csl,LL-DVD:19}. Many approaches
have been proposed to construct a CBF or verify whether a given
function is a CBF. One approach applies learning methods to construct
CBFs
\cite{KL-CQ-JC-NA:21-ral,CW-YM-YL-SLS-JL:21,MS-AD-SC-PAV:20,AR-HH-LL-HZ-DVD-ST-NM:20}. Another
uses reachability analysis to construct the maximal invariant set and
use it in safe control design \cite{JJC-DL-KS-CJT-SLH:21}. Another
class uses backstepping to design CBFs for cascaded
systems~\cite{AJT-PO-TGM-ADA:22}. Still, another group of works, most
closely related to the treatment here, utilizes properties of specific
systems to construct suitable CBFs. For instance,
\cite{AC:21a,MS-FD-SM:23} constructs CBFs for polynomial systems using
sum-of-square optimization.  The work~\cite{WSC-DVD:22} constructs
non-smooth CBFs for fully actuated Euler-Lagrange systems, with
constraints on position, velocity, and inputs given by hypercubes.
The work~\cite{AWF-TH:21} proposes a method to construct a safe subset
of a hyper-sphere in the position space, assuming no input
constraints. Both these works consider convex constraint sets. In
contrast, our treatment here is valid for
a highly expressive class of positional (potentially nonconvex)
constraints defined by linear boundaries.
%
%
Once a valid CBF is constructed, safe feedback controllers are usually
synthesized via state-parameterized optimization
programs~\cite{ADA-XX-JWG-PT:17,ADA-SC-ME-GN-KS-PT:19} due to their
flexibility, convenience, and computational lightness. This motivated
the study of the regularity properties of such controllers, see
e.g.~\cite{BJM-MJP-ADA:15,PM-AA-JC:25-ejc,MA-NA-JC:25-tac}.
%
%
Our recent work~\cite{MA-NA-JC:24-auto} synthesizes a provably
feasible optimization-based safe feedback controller for safe sets
given by arbitrarily nested unions and intersections of superlevel
sets of differentiable functions.

\emph{Statement of Contributions:} We consider\footnote{We use the
  following notation. 
  We let $\mathbf{1_{n\times m}}$ and $\mathbf{0_{n\times m}}$ denote
  the $n \times m$ matrices of ones and zeroes, resp. Likewise,
  $\mathbf{1_{n}}$ and $\mathbf{0_{n}}$ denote the $n$ vectors of ones
  and zeroes, resp.  The boundary, interior, and convex hull of $\Sc$
  are $\partial \Sc$, $\interior{\Sc}$, and $\co{\Sc}$, resp. The
  projection of $\Cc \subset \real^{n}$ on the first $n'$ components
  is denoted $\proj{n'}{\Cc}$.  The vectors of the standard basis of
  $\real^n$ are denoted $\{e_\ell\}_{\ell=1}^n$. We denote the 2-norm
  of a vector $x$ by $\|x\|$. The norm of a matrix $A$ is the induced
  2-norm, $\|A\| \coloneqq \sup_{x\neq 0} \frac{\|Ax\|}{\|x\|}$. We
  refer to the 0-superlevel set of a function simply as superlevel
  set.  } second-order system dynamics and positional constraints
specified by nested unions and intersections of half-spaces.  We
construct a control-invariant
subset of the full state space which contains all positions allowed by
the original positional constraints. We derive a general condition
which, if satisfied, proves that our constructed set is safe for
general, possibly not fully actuated, second-order systems and provide
an associated QP safeguarding controller.  We then show that this
safety condition is always satisfied for fully actuated systems.
We further show that a compact allowable controls set suffices for
safety when the desired safe set is bounded. We show how our method
can be utilized to incorporate velocity and input constraints when the
dynamics satisfy standard Euler-Lagrange system properties. Finally,
we apply our method to design safe controls for a 2D robotic arm.

\section{Problem Statement}\label{sec:prob}
Consider the second-order dynamics
\begin{align}
  \dot x = f(x) + G(x)u \label{eq:SODynamics}
\end{align}
where $x = (x_1,x_2)$, $x_1, x_2 \in \real^n$, $u \in \real^m$,
$f(x) = (x_2,f_2(x))$, $G(x) = (\mathbf{0}_{n \times m}, G_2(x))$,
with $f_2:\real^{2n} \to \real^n$ and
$G_2: \real^{2n} \to \real^{n \times m}$ Lipschitz functions. 
Consider also the half-spaces parameterized by
$i \in \Ic \coloneqq \{1,\dots,r\}$
\begin{align}
  \Cc_i \coloneqq \{x \in \real^{2n}\;|\; h_i(x) \coloneqq a_i^\top
  x_1 + b_i \geq 0\},\label{eq:ci} 
\end{align}
with $\|a_i\| \neq 0 \neq b_i$. We require that if $i\neq i'$ then the
augmented vectors $(a_i,b_i)$ and $(a_{i'},b_{i'})$ are linearly
independent. 
Our objective is to design a control law for \eqref{eq:SODynamics} that keeps invariant a desired set $\Cc$ given as a union of intersections of the half-spaces $\Cc_i$'s. That is,
\begin{align}
  \Cc &= \bigcup_{\ell \in \Lc} \bigcap_{i \in \Ic^\ell}\Cc_i, \label{eq:C}
\end{align}
where $\Lc \subset \mathbb{N}$
and the sets $\Ic^\ell \subset \Ic$
%
%
are sets of indices. This set corresponds to the superlevel set
%
%
of the
function
\begin{align}
  h(x) = \max_{\ell \in \Lc} \min_{i \in \Ic^\ell} h_i(x). \label{eq:h}
\end{align}
Note that the safety set $\Cc$ constrains only the states~$x_1$
corresponding to the generalized position. This setup is common in
many problems, such as collision avoidance~\cite{AS-KK-JB-AB-PT-AA:21},
where the main concern is to avoid the physical locations occupied by
obstacles.
%
This form for $\Cc$ is flexible enough to capture the safety
requirement of staying in any set of positions with linear boundaries while
simultaneously avoiding obstacles of linear boundaries. This is applicable in a wide range of situations for autonomous robotic systems from geofencing to protect a human collaborator to avoiding collisions in human environments.

We impose the following structural assumptions on the safety
requirement.  Let $\Sc_{\cap}$ be the collection of sets of indices
whose corresponding half-spaces intersect in $\Cc$: that is, if
$I \subseteq \Ic$ is such that
%
%
$\left ( \cap_{i \in I} \Cc_i \right ) \cap \Cc \neq \emptyset$, then
$I \in \Sc_\cap$.
%
%

\begin{assumption}\label{as:bounds}
  {\rm The set
    $\proj{n}{\Cc}$
    %
    is compact. Furthermore, for any $I \in \Sc_\cap$, one point
    $y_I \in \Cc$
    can be chosen to satisfy $h_i(y_I) > 0$ for all
    $i \in I$.  } \oprocend
\end{assumption}

Assumption~\ref{as:bounds} is reasonable.
%
Assuming the compactness of $\proj{n}{\Cc}$ is the same as saying that
the set of safe positions $x_1$ is compact.
%
%
A sufficient condition for this is the boundedness of
$\cap_{i \in \Ic^\ell}\Cc_i$ for every $\ell$. The last part of the
assumption only requires that any nonempty intersection
$\left ( \cap_{i \in I}\Cc_i \right ) \cap \Cc$ has a non-empty
interior. We discuss the impact of the specific choice of points
$\{y_I\}$ in Remark~\ref{rem:OI}.

Since our system \eqref{eq:SODynamics} is second order, the control is
only available in the order of the generalized acceleration
$\dot x_2$. Thus, $\Cc$, which only constrains~$x_1$, is generally not
control-invariant due to the lack of constraints on the velocity:
e.g., initial conditions starting exactly at the boundary of $\Cc$
with a velocity heading outwards are unsafe, with no possible control
value to counter it. Hence, there is a need to identify a
control-invariant set $\Cc^s \subseteq \Cc$ that constrains the
velocity. This construction should contain as much of $\proj{n}{\Cc}$
as possible.
%
We thus formulate our problem as follows.

\begin{problem}\label{prob:problem}
  \rm{Given the second-order dynamics~\eqref{eq:SODynamics} and the
    set $\Cc$ defined by~\eqref{eq:C} that satisfies Assumption~\ref{as:bounds}, construct a set $\Cc^s$ such
    that:
    \begin{enumerate}
    \item $\Cc^s \subseteq \Cc$,
    \item $\proj{n}{\Cc^s} = \proj{n}{\Cc}$,
      %
      %
    \item $\Cc^s$ is control-invariant,
    \end{enumerate}
    and design a continuous controller that renders $\Cc^s$
    invariant.}
\end{problem}

\section{Construction of Control-Invariant Set}\label{sec:cbarConstruction}

In this section, we solve Problem~\ref{prob:problem} by constructing a
function $B$ of the form
\begin{align*}
  B(x) = \max_{\ell \in \bar \Lc} \min_{i \in \bar \Ic^\ell} B_i(x) ,
\end{align*}
and then proving that its superlevel set $\Cc^s$
satisfies the requirements in Problem~\ref{prob:problem}. 
Constructing $B$ amounts to defining the functions $B_i$ and the index sets $\bar \Lc$ and $\bar \Ic^\ell$.

\textbf{Definition of functions:} For each $1 \leq i \leq r$, define 	$B_i(x) = h_i(x)$ and
  \begin{equation}\label{eq:bs}
    B_{i+r}(x) = a_i^\top x_2 + \gamma(a_i^\top x_1 + b_i)-\epsilon,
  \end{equation}
  where $\epsilon$ and
  $\gamma$ are two positive design parameters, each of which will have a special role in customizing the design and proving its safety.
  We denote by $\Cc^s_i$ the superlevel set of $B_i$.

\textbf{Definition of sets:} Let $\bar \Lc = \Lc$ and
  $\bar \Ic^\ell = \Ic^\ell \cup (\{r\}+\Ic^\ell)$.

The choice of $B$ follows this logic: for $i \in \until{r}$, $\Cc_i$
is the superlevel set of $B_i$; the function $B_{i+r}$ is then chosen
such that its superlevel set contains the points on the boundary of
$\Cc_i$ only if the system drift $f$ points towards the interior of
the safe set at those boundary points. Thus, the system is safe
without requiring any inputs at the boundary points defined by $B_i=0$
with $i \in \until{r}$ if these boundary points are in the superlevel
set of $B_{i+r}$. This choice of $B_{i+r}$ is inspired by the concept
of high-order control barrier function~\cite{WX-CB:22}.

The following result shows that this construction satisfies requirements (i) and (ii) in Problem~\ref{prob:problem}.

\begin{lemma}[No Positions Lost]\label{lem:noLostPos}
  Let $\Cc^s$ be the superlevel set of $B$. Then $\Cc^s \subset \Cc$.
  Under Assumption~\ref{as:bounds}, define
  \begin{align}\label{eq:delta}
    \delta \coloneqq \min_{I \in \Sc_\cap, \; i \in I} h_i(y_I) >0,
  \end{align}
  and let $\gamma, \epsilon$ with $\gamma \delta > \epsilon$. Then,
  $\proj{n}{\Cc^s} = \proj{n}{\Cc}$.  \oprocend
\end{lemma}


\begin{proof}
  That $\Cc^s$ is a subset of $\Cc$ follows directly from the
  definition of these sets and the fact that $\Cc^s_i= \Cc_i$ for
  $i \in \until{r}$.  This also implies that
  $\proj{n}{\Cc^s} \subset \proj{n}{\Cc}$.
  Note that Assumption~\ref{as:bounds} ensures that $\delta >
  0$. Since by assumption $\gamma \delta > \epsilon$, there exists
  $0<\sigma<1$ such that $\gamma \sigma \delta > \epsilon$.  Given
  $(x_1,x_2) \in \Cc = \cup_{\ell \in \Lc} \cap_{i \in
    \Ic^\ell}\Cc_i$, there is an $\ell' \in \Lc$ such that
  $(x_1, x_2) \in \Cc_i$, for all $i \in \Ic^{\ell'}$, i.e.,
  $h_i(x) = B_i(x) \geq 0$, for all $i \in \Ic^{\ell'}$. By
  Assumption~\ref{as:bounds}, there exists $y_{\Ic^{\ell'}} \in \Cc$
  satisfying $h_i(y_{\Ic^{\ell'}}) > 0$ for all $i \in \Ic^{\ell'}$.
  %
  %
  The choice $x' = (x_1,x_2')$, where $x_2' = -\gamma \sigma (x_1 - y_{\Ic^{\ell'}})$
  gives
  \begin{align*}
    B_{i+r}(x')
    &= a_i^\top x_2' + \gamma B_i(x) - \epsilon
    \\
    & = \gamma (1-\sigma)B_i(x) + \gamma \sigma B_i(y_{\Ic^{\ell'}})
      - \epsilon
    \\
    & \geq \gamma (1-\sigma)B_i(x) + \gamma \sigma \delta - \epsilon
    \\
    & > \gamma (1-\sigma)B_i(x) \geq 0.
  \end{align*}
  Therefore, $B_i(x') \geq 0$ for all $i \in \bar \Ic^{\ell'}$,
  implying that
  $x' \in \cap_{i \in \bar \Ic^{\ell'}}\Cc^s_i \subseteq
  \Cc^s$. Therefore, $\proj{n}{\Cc} \subset \proj{n}{\Cc^s}$. 
\end{proof}

\begin{remark}[Maximizing Flexibility of Safe Set
  Design]\label{rem:OI}
  \rm{Lemma~\ref{lem:noLostPos} requires that the parameters
    $\gamma$ and $\epsilon$ satisfy $\gamma \delta >
    \epsilon$. Since $\delta$ is dependent on the choice of points
    $\{y_I\}$, whose existence is assumed in
    Assumption~\ref{as:bounds}, the choice of $\gamma$ and
    $\epsilon$ is dependent on the set of points $\{y_I\}$. The
    greater $\delta$, the more flexibility for choosing $\gamma$ and
    $\epsilon$.
    Therefore, for a set $I \in \Sc_\cap$, the choice of $y_I$ that
    promotes the most design flexibility is
    $y_I \coloneqq \argmax_{x \in \Cc} \min_{i \in I} h_i(x)$. It is
    important to note, however, that the specific choice of points
    $\{y_I\}$, while important for design flexibility, is not
    crucial for the results of this paper: any choice makes the
    condition $\gamma \delta > \epsilon$ satisfiable and thus allows
    for the design of a safe~$\Cc^s$ satisfying
    $\proj{n}{\Cc^s} = \proj{n}{\Cc}$.} \oprocend
\end{remark}

To address requirement (iii) in Problem~\ref{prob:problem}, we need to
identify a control action at each state of $\Cc^s$ that makes $\Cc^s$
forward-invariant. We introduce the following condition.

\begin{condition}[General Safety Condition]\label{as:cbfCond}
  {\rm For each $x \in \partial \Cc^s$, there exists $u_x \in \Uc$
    such that
    \begin{align*}
      \dot B_i(x) = \nabla B_i(x)^\top(f(x)+G(x)u_x) > 0
    \end{align*}
    for all
    $i \in \tilde \Ic(x) \coloneqq \{i \in \until{2r} \;|\; \exists
    \ell \in \Lc_i: \; B(x) = B^\ell(x) = B_i(x)\}$, with
    $B^\ell(x) = \min_{i \in \bar \Ic^\ell} B_i(x)$ and
    $\Lc_i \coloneqq \{\ell \in \bar \Lc \;|\; i \in \bar \Ic^\ell
    \}$.} \oprocend
\end{condition}
%
%
Condition~\ref{as:cbfCond} requires that there exists a control input that steers the system to the interior of the safe set at all boundary points where the drift of the system is not guaranteed to to do that on its own.

Consider now the feedback controller $u^*$ defined by the following quadratic program:
\begin{align}
  (u^*&(x),\alpha^*(x),M^*(x)) \coloneqq \notag
  \\
      & \hskip -5mm \argmin_{\alpha, M, u \in \Uc} {u^\top Q(x)u +
        q(x)^\top u + q_\alpha \alpha^2 + q_M M^2} \label{eq:qp}
  \\
  \text{s.t. } & M \geq c_M, \; \alpha \geq c_\alpha \nonumber
  \\
      & \nabla B_i(x)^\top (f(x) + G(x)u) + \alpha B_i(x) \nonumber
  \\
      &+ M\left (B(x) - B^\ell(x) \right ) \geq 0, \nonumber \;
        \forall \ell \in \bar \Lc , \; \forall i \in \bar \Ic^\ell.
        \nonumber 
\end{align}
Here, $Q:\Xc \to \real^{m \times m}$ is a Lipschitz function which
takes values in the set of positive-definite matrices,
%
%
$q:\Xc \to \real^m$ is a Lipschitz function, and $q_\alpha$, $q_M$,
$c_M$, $c_\alpha$ are positive design constants.  The following result
summarizes the properties of~$u^*$ and, in particular, that it makes
$\Cc^s$ control-invariant.
under Condition~\ref{as:cbfCond}.

\begin{theorem}\longthmtitle{Safe
    Controller~\cite[Thm. 4.12]{MA-NA-JC:24-auto}}\label{thm:generalControlInvariance} 
  Let $\Cc^s$ be compact.
  Under Condition~\ref{as:cbfCond}, there exists a neighborhood of
  $\Cc^s$ where program~\eqref{eq:qp} is feasible and $u^*$ is
  single-valued, continuous, and renders $\Cc^s$
  control-invariant. \oprocend
  %
  %
\end{theorem}

Our approach to establish (iii) in Problem~\ref{prob:problem} is then
to verify the hypotheses of
Theorem~\ref{thm:generalControlInvariance}.  The following result
shows that $\Cc^s$ is compact if $\proj{n}{\Cc}$ is compact,
which is readily ensured by Assumption~\ref{as:bounds}.

\begin{proposition}[Compactness of Safe Set]
  \label{prop:compactCbar}
  Under Assumption~\ref{as:bounds}, $\Cc^s$ is compact. \oprocend
\end{proposition}
\begin{proof}
  We reason by contradiction. Suppose
  $\Cc^s = \cup_{\ell \in \bar \Lc} \cap_{i \in \bar \Ic^\ell}\Cc^s_i$
  is not compact. Then, for some $\ell' \in \Lc$, the convex closed set
  $\cap_{i \in \bar \Ic^{\ell'}}\Cc^s_i$ is unbounded and, thus,
  contains a ray \cite[Result 2.5.1]{BG-VK-MAP-GCS:03}.
  That is there is a
  point $x_0 = (x_{0,1},x_{0,2})$ and a direction
  $\zeta = (\zeta_1, \zeta_2)\neq \mathbf{0}$ such that
  $B_i(x_0 + t \zeta) \geq 0$ for all $t \geq 0$ and all
  $i \in \bar \Ic^{\ell'}$.  We distinguish two cases: (a)
  $\zeta_1 \neq \mathbf{0}$ and (b) $\zeta_1 = \mathbf{0}$. In
  case~(a), $B_i(x_0 + t \zeta) = h_i(x_0 + t \zeta)\geq 0$ for all
  $i \in \Ic^{\ell'}$, so $\proj{n}{\Cc}$ is not compact, which
  contradicts Assumption~\ref{as:bounds}.  In case~(b),
  $\zeta_1 = \mathbf{0} \neq \zeta_2$. So,
  $B_{i+r}(x_0 + t \zeta) = t a_i^\top \zeta_2 + a_i^\top x_{0,2} +
  \gamma (a_i^\top x_{0,1} + b_i) - \epsilon \geq 0$, for all
  $i \in \Ic^{\ell'}$ and all $t \geq 0$. Since all terms in the
  inequality are constants except for the first one, we deduce that
  $t a_i^\top \zeta_2 \geq 0$ for $t$ large enough, which implies
  $a_i^\top \zeta_2 \geq 0$, for all $i \in \Ic^{\ell'}$.  Therefore,
  $h_i(x_0 + t(\zeta_2,\mathbf{0})) = t a_i^\top \zeta_2 + a_i^\top
  x_{0,1}+ b_i = t a_i^\top \zeta_2 + h_i(x_0) \geq 0$ for all
  $t \ge 0$ and all $i \in \Ic^{\ell'}$, which implies that
  $\proj{n}{\Cc}$ is unbounded, contradicting
  Assumption~\ref{as:bounds}.
\end{proof}

Next, we focus on the satisfaction of Condition~\ref{as:cbfCond}. The
following result particularizes this condition to our context.

\begin{lemma}[General Invariance Condition]\label{lem:2ndOrderCond}
  If for all $x \in \partial \Cc^s$ and $i'+r \in \tilde \Ic(x)$
  for which $B_{i'+r}(x) = 0$, there exists $u_x \in \Uc$ such that
  \begin{align}\label{eq:2ndOrderCond}
    a_{i'}^\top (\gamma x_2+f_2(x)+G_2(x)u_x)  > 0, 
  \end{align}
  %
  %
  then Condition~\ref{as:cbfCond} is satisfied. \oprocend
\end{lemma}

\begin{proof}
  If $x \in \partial \Cc^s$, then $B(x) = 0$ and thus $B_i(x) = 0$ for
  $i \in \tilde \Ic(x)$. For any such $i \in \tilde \Ic(x)$, it
  follows that $B_{i'}(x) \geq 0$ for any $\ell \in \Lc_i$ and
  $i' \in \bar \Ic^\ell$. Let us now verify the inequality of
  Condition~\ref{as:cbfCond} for $i \in \tilde \Ic(x)$. We distinguish
  two cases: (a) $i \leq r$ or (b) $i>r$. In case (a), $B_i(x) = a_i^\top x_1 +b_i = 0$ and $i' = i+r \in \bar
  \Ic^\ell$. Thus,
  \begin{align}
    0 \leq B_{i'}(x)
    &= a_i^\top x_2 + \gamma B_i(x) -\epsilon=a_i^\top x_2 - \epsilon. \label{eq:x2}
  \end{align}
  So $\dot B_i(x) = a_i^\top x_2 \geq \epsilon > 0$
  by~\eqref{eq:x2}. In case (b), for $i' = i - r$,
  $\dot B_i(x) = \dot B_{i'+r}(x)$ equals the left-hand side
  of~\eqref{eq:2ndOrderCond}, which verifies
  Condition~\ref{as:cbfCond} by assumption.
\end{proof}

We use the inequality~\eqref{eq:2ndOrderCond} to establish safety
under specific structural assumptions on our dynamics.  The inequality
can fail to hold in two ways. The first is when there is no
$u_x \in \real^m$ that satisfies it. The second is when such a
$u_x \in \real^m$ exists, but it does not belong to the allowable
input set~$\Uc$. The distinction between the two possibilities is
useful since, in practice, they can be overcome by different tools. If
the failure is of the second type, then a practical solution might be
to expand $\Uc$ (e.g., by employing a more powerful actuator). If the
failure is of the first type, such a solution is not possible, but
rather the set $\Cc^s$ cannot be made safe given the dynamics, and a
better design must be sought. We exploit this distinction in what
follows.

The following result shows that Condition~\ref{as:cbfCond} is satisfied by fully
actuated systems.

\begin{proposition}[Control-Invariance with Full
  Actuation]\label{prop:CIwFullControl}
  Assume that, for all $x \in \partial \Cc^s$, $G_2(x)$ is right
  invertible and $\Uc = \real^m$. Then, under
  Assumption~\ref{as:bounds}, choosing $\gamma$ and $\epsilon$ such
  that $\gamma \delta > \epsilon$ ensures that
  Condition~\ref{as:cbfCond} is satisfied.  \oprocend
\end{proposition}
\begin{proof}  
  Given $x = (x_1,x_2) \in \partial \Cc^s$, $B(x) = 0$. Define
  $I_x \coloneqq \{i \in \{1,\dots, r\} \;|\; i+r \in \tilde \Ic(x),
  B_{i+r}(x) = 0\}$. If $I_x$ is empty then the premise of
  Lemma~\ref{lem:2ndOrderCond} is satisfied by default. Otherwise,
  by the definitions of $\tilde \Ic(x)$ and $\bar \Ic^\ell$'s,
  $B_i(x) \geq 0$ for all $i \in I_x$.  Therefore, by
  Assumption~\ref{as:bounds}, there exists $y_{I_x} \in \Cc$
  satisfying $h_i(y_{I_x})>0$ for all $i \in I_x$.
  %
  Then, for all $i \in I_x$
  \begin{align}
    B_{i+r}(x) = a_i^\top x_2 + \gamma(a_i^\top x_1
    &+ b_i) - \epsilon
      = 0 \notag
    \\ 
    \implies a_i^\top (-x_2 - \gamma x_1 +\gamma y_{I_x})
    &=
      \gamma(a_i^\top
      y_{I_x} +
      b_i) -
      \epsilon
      \notag
    \\
    & \geq \gamma \delta - \epsilon > 0. \label{eq:yx}
  \end{align}
  Now, because of~\eqref{eq:yx}, one can choose
  $u_x = \rho G_2^\dagger (x)y_x$, where $G_2^\dagger $ denotes the right
  inverse of $G_2$ and $y_x = -x_2 - \gamma x_1 +\gamma y_{I_x}$, with
  $\rho$ large enough, so that the inequality~\eqref{eq:2ndOrderCond}
  is satisfied. The result then follows from~Lemma~\ref{lem:2ndOrderCond}.
\end{proof}

The combination of Proposition~\ref{prop:compactCbar} and
Proposition~\ref{prop:CIwFullControl} allows us to invoke
Theorem~\ref{thm:generalControlInvariance} to establish that $u^*$, as
defined in~\eqref{eq:qp}, renders $\Cc^s \subseteq \Cc$
forward-invariant. This, together with Lemma~\ref{lem:noLostPos},
means that $\Cc^s$ solves Problem~\ref{prob:problem} for the case of
full actuation.

The assumption of full actuation is not uncommon in the control of
second-order systems, whether when studying
safety~\cite{WSC-DVD:22,AWF-TH:21} or
stability~\cite{RO-AL-PJN-HSR:98}.  The following result shows that
the result of Proposition~\ref{prop:CIwFullControl} still holds for
sufficiently large compact input sets~$\Uc$.

\begin{corollary}[Compact Input Set Suffices for
  Control-Invariance]\label{cor:CIwFullControl}
  Assume that, for all $x \in \partial \Cc^s$, $G_2(x)$ is right
  invertible. Then, under Assumption~\ref{as:bounds}, choosing
  $\gamma$ and $\epsilon$ such that $\gamma \delta > \epsilon$ ensures
  that Condition~\ref{as:cbfCond} is satisfied for some compact input
  set $\Uc \subsetneq \real^m$. \oprocend
\end{corollary}
%
\begin{proof}
  Since, by Proposition~\ref{prop:compactCbar}, $\partial \Cc^s$ is
  compact and the left-hand side of \eqref{eq:2ndOrderCond} is
  continuous in $x$, $a_i^\top (\gamma x_2 + f_2(x))$ is bounded in
  $\partial \Cc^s$. But
  $a_i^\top G_2(x)G_2(x)^\dagger y_x = a_i^\top y_x \geq \gamma \delta -
  \epsilon > 0$. So there is a finite $\rho$ that
  validates~\eqref{eq:2ndOrderCond} for all $x$ with
  $u_x = \rho G_2(x)^\dagger y_x$. Noting the boundedness of $y_x$ and
  $G_2(x)^\dagger $ in $\partial \Cc^s$ completes the proof.
\end{proof}

This result does not specify how large the input set should be. This
is the task we tackle in the forthcoming section.


\section{Determining Input Magnitude for Control Invariance of
  Euler-Lagrange Systems}\label{sec:inputConstraints}
%

%


In this section we consider the class of Euler-Lagrange
systems~\cite{RO-AL-PJN-HSR:98} and study how large the input set should be 
%
%
to render $\Cc^s$ control-invariant. 

\begin{assumption}[Input Set Structure and Euler-Lagrange Systems
  Properties]\label{as:stop}
  \rm{Let $\Uc = \{u \in \real^m \;|\; \|u\| \leq d\}$ for some
    $d > 0$. Further, assume the dynamics~\eqref{eq:SODynamics} is
    such that:
    \begin{enumerate}[(a)]
    \item the matrix function $G_2(x)$ is only dependent on $x_1$ and
      has right inverse $G_2^\dagger (x_1)$ defined for all
      $x \in \Cc$;
      %
      %
      and
    \item $f_2(x) = f^1_2(x_1) + f^2_2(x)$, with
      $\|f^2_2(x)\| \leq k_2 \|x_2\|$. \oprocend
  \end{enumerate}}
\end{assumption}

Assumption~\ref{as:stop}(a) can be interpreted as having the inertia
matrix $G_2^\dagger (x)$ solely depend on the system's positions (and
not on the system's velocity). Assumption~\ref{as:stop}(b) splits the
forcing on the system into a component generated by potential fields,
such as gravity, $f_2^1(x_1)$, and a component containing other
forces, such as damping, $f_2^2(x)$. It also asks that the magnitude
of the latter to be at most proportional to the system's
velocity~$\|x_2\|$. Define constants $k_1, k_G$ by
\begin{align}
  k_1 \coloneqq \max_{x \in \Cc} \|f_2^1(x_1)\| \;\text{  and  }\;
  k_G \coloneqq \max_{x \in \Cc} \|G_2^\dagger (x_1)\|. \label{eq:ks} 
\end{align}
The existence of these constants is ensured by the continuity of $f_2$
and $G_2^\dagger $ and the compactness of $\proj{n}{\Cc}$,
cf.~\cite[Thm. 4.16]{WR:76}. Our approach to establishing
control-invariance under Assumption~\ref{as:stop} has two steps.
First, we show that the velocity magnitude $\|x_2\|$ in the safe set
$\Cc^s$ can be arbitrarily bounded by the design parameter $\gamma$,
cf. Lemma~\ref{lem:gamma}. Second, we show that $\Uc$ in
Assumption~\ref{as:stop} is sufficient for control-invariance when
$\|x_2\|$ is forced to be small enough through suitable design of
$\gamma$, cf. Theorem~\ref{thm:limitedU}.

\begin{lemma}[Bound on Velocity Magnitude in Safe
  Set]\label{lem:gamma}
  Under Assumption~\ref{as:bounds}, there exists a constant $c$ that
  depends only on $\{a_i,b_i\}_{i \in \Ic}$ defining
  $\{h_i\}_{i \in \Ic}$ such that, for all $\gamma$ and $\epsilon$
  satisfying $\gamma \delta > \epsilon$, $\|x_2\| \leq \gamma c$ for
  all $x = (x_1,x_2) \in \Cc^s$. \oprocend
\end{lemma}
%
%
\begin{proof}
  By Proposition~\ref{prop:compactCbar}, $\Cc^s$ is
  bounded.
  %
  Therefore, each of its (finitely many) components
  $\cap_{i \in \bar \Ic^\ell} \Cc^s_i$ is bounded too.  Since each
  $\Cc^s_i$ is a half-space, $\cap_{i \in \bar \Ic^\ell} \Cc^s_i$ is a
  bounded polytope. Given $\ell \in \bar \Lc$, consider the $n$
  programs
  \begin{align*}
    x_j^*
    &= \argmax |e_j^\top x_2|
    \\
    & \qquad \text{s.t.}\; x \in \cap_{i \in \bar \Ic^\ell} \Cc^s_i,
  \end{align*}
  where $j \in \{1,\dots, n\}$. Let
  $j_\ell^* \coloneqq \argmax_{1\leq j \leq n} \{|e_j^\top
  x_{j,2}^*|\}$, where $x_{j,2}^*$ denotes the last $n$ components of
  $x_{j}^*$,
  %
  %
  and $x^{*,\ell} \coloneqq x^*_{j_\ell^*}$. Since
  $x^{*,\ell} \in \real^{2n}$ is a solution to a linear program over a
  polytope, it is a vertex of the
  polytope~\cite[Thm. 2.4]{DG-MJT:89}. By~\cite[Thm. 10.4]{AB:12},
  there are $2n$ indices
  %
  %
  %
  $\Ic_v \subseteq \bar \Ic^\ell$ such that $B_i(x^{*,\ell})=0$, for
  all $i \in \Ic_v$. Direct evaluation yields
  \begin{align*}
    a_i^\top
    \begin{bmatrix}
      I
      &
        \pmb{0}_{n \times n}
    \end{bmatrix}
    x^{*,\ell}
    &= -b_i, \; \text{if} \; i \in \Ic_v \cap \Ic,
    \\ 
    a_i^\top
    \begin{bmatrix}
      \gamma I & I
    \end{bmatrix}
    x^{*,\ell}
    &=
      -\gamma
      b_i
      +
      \epsilon,
      \;
      \text{if}
      \;
      i+r
      \in
      \Ic_v. 
  \end{align*}
  Stacking the above equations into one matrix equation gives
  \begin{align}
    \begin{bmatrix}
      A_1 & \pmb{0}\\
      \gamma A_2 & A_2
    \end{bmatrix} x^{*,\ell} = \begin{bmatrix}
      b_1 \\ \gamma b_2 + \epsilon \mathbf{1_n}
    \end{bmatrix}, \label{eq:asbs}
  \end{align}
  with appropriate matrices $A_1, A_2, b_1$ and $b_2$.
  Since $x^{*,\ell}$ is a unique solution as it is a vertex,
  the coefficient matrix
  in~\eqref{eq:asbs} is invertible, and thus is of rank $2n$. The
  matrices $A_1$ and $A_2$ have $n$ columns, so their ranks are at
  most $n$. Since the rank of the block
  $\begin{bmatrix} \gamma A_2 & A_2 \end{bmatrix}$ is equal to the
  rank of $A_2$, then it must be that both rank$(A_1)=$
  rank$(A_2)= n$, and hence and both $A_1$ and $A_2$ are invertible.
  Solving for $x_2^{*,\ell}$, which is the vector comprising the last
  $n$ components of $x^{*,\ell}$ gives
  $x_2^{*,\ell} = \gamma (A_2^{-1}b_2 - A_1^{-1}b_1) + \epsilon
  A_2^{-1}\mathbf{1_n}$.
  %
  By definition of $x^{*,\ell}$ and $j_\ell^*$
  and the triangle inequality, for all $j \in \{1, \dots, n\}$ and all
  $x = (x_1,x_2) \in \cap_{i\in \bar \Ic^\ell} \Cc^s_i$,
  \begin{align*}
    |e_j^\top x_2| \leq |e_{j_\ell^*}^\top x_2^{*,\ell}| \leq
    \gamma c'_1 + \epsilon c'_2 \leq \gamma (c'_1 + \delta c'_2) ,    
  \end{align*}
  where $c'_1 = |e_{j_\ell^*}^\top (A_1^{-1}b_1 - A_2^{-1}b_2)|$ and
  $c'_2 = \|A_2^{-1}\mathbf{1_n}\|$.
  Since this holds for each
  $\ell \in \bar \Lc$
  with the appropriate matrices $A_1,A_2,b_1$ and
  $b_2$, every component of
  $x_2 \in \Cc^s = \cup_{j \in \bar \Lc} \cap_{i \in \bar \Ic^\ell}
  \Cc^s_i$ is bounded by $\gamma c'^*$, where $c'^*$ is the greatest
  of the finite constants $c'_1+\delta c'_2$ for the different
  $\ell$'s. Therefore, $\|x_2\| \leq \gamma (\sqrt{n}c'^*)$.
  %
  %
\end{proof}

An interesting byproduct of Lemma~\ref{lem:gamma}
is that the design parameter $\gamma$ can be used to ensure safety
constraints on the velocity too. In fact, one can leverage the result
to meet any safety specification on the magnitude of $x_2$ by taking a
sufficiently small~$\gamma$. The smaller the parameter~$\gamma$ is,
the slower the system will move, and therefore control objectives
other than safety may be hindered.
%
%
Notice also that, no matter how small $\gamma$ is taken to be, with a
suitably small $\epsilon$, none of the originally allowed positions in
$\Cc$ will be lost in $\Cc^s$, cf. Lemma~\ref{lem:noLostPos}.
%
%
We are now ready to show that Assumption~\ref{as:stop} is enough to
establish the existence of a design parameter $\gamma$ that makes
$\Cc^s$ control-invariant with bounded input set
$\Uc = \{u \in \real^m \;|\; \|u\| \leq d\}$.

\begin{theorem}[Control-Invariance with Input
  Constraints]\label{thm:limitedU}
  Under Assumptions~\ref{as:bounds} and~\ref{as:stop}, if
  $d - k_G k_1 > 0$, cf.~\eqref{eq:ks},
  then $\gamma$ and $\epsilon$ can be chosen to ensure
  Condition~\ref{as:cbfCond} is satisfied. \oprocend
\end{theorem}

\begin{proof}
  Let $\gamma$ be such that
  $\gamma (k_2+\gamma)k_G c < \frac{1}{2}(d-k_1k_G)$, where $c$ is the
  constant in Lemma~\ref{lem:gamma}. Now, select $\epsilon$ to satisfy
  $\gamma \delta > \epsilon$.
  %
  %
  For all $x \in \partial \Cc^s$
  and $i \in \tilde \Ic(x)$ with $i > r$, recall from the proof of
  Proposition~\ref{prop:CIwFullControl} that $a_i^\top y_x > 0$, where
  $y_x$ is as defined there. Let $\beta_x$ be a positive constant
  satisfying $\beta_x \|y_x\| \leq \frac{1}{2}(d - k_1k_G)$ and choose
  %
  %
  $u_x = -G_2^\dagger (x)(f_2(x)+\gamma x_2) + \beta_x y_x$. By the triangle
  inequality and the definition of the matrix induced 2-norm,
  \begin{align*}
    \|u_x\| \leq \|G_2^\dagger (x)\|(\|f_2^1(x_1)\|+\|f_2^2(x)\|+\gamma \|x_2\|)
    + \beta_x\|y_x\| .    
  \end{align*}
  By Assumption~\ref{as:stop}, Lemma~\ref{lem:gamma},
  and~\eqref{eq:ks},
  $\|u_x\|\leq k_Gk_1 + \gamma (k_2+\gamma)k_G c + \beta_x\|y_x\|$. By
  our choice of $\gamma$ and $\beta_x$,
  $\|u_x\| \leq k_Gk_1 + \frac{1}{2}(d-k_Gk_1) + \frac{1}{2}(d-k_Gk_1)
  = d$, and thus $u_x \in \Uc$. Noting that $u_x \in \Uc$
  validates~\eqref{eq:2ndOrderCond} for all $i \in \tilde \Ic(x)$ with
  $i > r$ completes the proof.
\end{proof}

The condition $d - k_Gk_1 > 0$ in Theorem~\ref{thm:limitedU} amounts
to having enough control authority to counter the force caused by the
potential field. In the absence of such a force (i.e., when $k_1 = 0$),
as in the case of a
robot moving in a plane perpendicular to gravity, this condition is
valid by default. For such systems, $\Cc^s$ can be made safe for any
bounded~$\Uc$.


\section{Safe Control of Robotic Manipulator}
%
%
%
In this section, we illustrate our results on a 2-degree-of-freedom
planar elbow manipulator \cite{MWS-SH-MV:20}. We build upon
\cite{WSC-DVD:22}, which consider the same example with hyper-cubic
safety constraints on the position, and generalize these to general
polytopic constraints. The system consists of two horizontally
oriented links (no gravity) hinged to each other, with the first link
hinged to a fixed frame. Torque manipulation is available at each
joint. The system model is
\begin{align}\label{eq:ex1Dynamics}
  \underbrace{
  \begin{bmatrix}
    c_{11} + c_{12}\cos(\theta_2)
    & c_{13} + c_{14}\cos(\theta_2)
    \\
    c_{22} + c_{23}\cos(\theta_2)
    & c_{21}
  \end{bmatrix}}_{M}
  \begin{bmatrix}
    \ddot \theta_1
    \\
    \ddot \theta_2
  \end{bmatrix}
  & =
  \\
    & \hskip -4cm 
      \underbrace{
      \begin{bmatrix}
        c_{15}\sin(\theta_2)\dot \theta_2
        & c_{16}\sin(\theta_2)\dot \theta_2
        \\
        c_{24}\sin(\theta_2)\dot \theta_1
        & 0
      \end{bmatrix}}_{C}
      \begin{bmatrix}
        \dot \theta_1
        \\
        \dot \theta_2
      \end{bmatrix}
      \notag
  \\
    & \hskip -2.5cm
      +\begin{bmatrix}
         u_1
         \\
         c_{25}\cos(\theta_1 + \theta_2) + u_2
       \end{bmatrix},
      \notag
\end{align}
with constant coefficients $c_{ij}$ dependent on the links' lengths
and masses, cf.~\cite{MWS-SH-MV:20}. We require the arm to avoid
collision with a wall, as shown in Figure~\ref{fig:wall}.
%
%

In the space of angles, the positional constraints $\Cc$ correspond to
the dotted hexagon in Figure~\ref{fig:portrait}. Thus, $y_I$ can be
taken as the origin for all the intersections. The parameter $\delta$
can be calculated according to~\eqref{eq:delta} to be $\delta = \pi/2$.
%
%
We construct the control-invariant set $\Cc^s$ as described in
Section~\ref{sec:cbarConstruction}, with $\gamma = 10$ and
$\epsilon = 0.1$, which satisfy $\gamma \delta > \epsilon$.
%
%
The controller~$u^*$ is computed according to~\eqref{eq:qp}, with the
objective function $\|u - u_{\text{nom}}(t)\|^2$. Here,
$u_{\text{nom}}$ is a nominal controller that tracks
$r = (\pi \sin(t),\frac{\pi}{2}\sin(4t))$, cf.~\cite{WSC-DVD:20},
\begin{align*}
  u_{\text{nom}}(t) = M(\ddot r -\dot e - e) + C [\dot \theta_1 \;\;
  \dot \theta_2]^\top ,  
\end{align*}
where $e = (\theta_1,\theta_2) - r$ and $M$ and $C$ are as defined
in~\eqref{eq:ex1Dynamics}. We take $\alpha(r) = 40r$. The high value
$40$ is not necessary for feasibility of the program~\eqref{eq:qp},
but we select it to allow the trajectory to come close to the boundary
of the safe set.


\begin{figure*}
  \begin{subfigure}{0.15\linewidth}
  \centering
  \begin{tikzpicture}[scale=.95]
    \begin{axis}[xmin=-.2, xmax=0.6,ymin=-1,ymax=2,xtick=\empty,
      ytick=\empty,width=3.5cm, height=5.6cm]
      \draw[draw = gray, fill = gray] (axis cs: 0,-3) rectangle (axis cs: -.2,3);
      \draw [ultra thick] (axis cs: 0,0) node{}
      -- (axis cs: 0.2,0.6) node{}
      -- (axis cs: 0.3,1.4) node{};
      \draw [thick, dotted] (axis cs: 0,0) node{}
      -- (axis cs: 0.2,0) node{};
      \draw [thick, dotted] (axis cs: 0.2,0.6) node{}
      -- (axis cs: 0.4,0.6) node{};
      \draw (axis cs: 0.2,0.12) node {$\theta_1$};
      \draw (axis cs: 0.4,0.75) node {$\theta_2$};
    \end{axis}
  \end{tikzpicture}
  \caption{A schematic of the 2-degrees-of-freedom robot arm. The gray
    space is the wall that should be avoided during manipulation.}
  \label{fig:wall}
\end{subfigure}
\hspace{.01mm} 
  \begin{subfigure}{0.25\linewidth}
    \centering
    \begin{tikzpicture}
      \begin{axis}[ylabel style={yshift=-2.5mm}, label style={font=\footnotesize}, tick label style={
        font=\footnotesize}, height=3.5cm, width=5.3cm, xmajorticks=false,ymax = 6,legend columns=-1,legend pos = north west, legend style={draw=none,font=\tiny}, ylabel = {$\theta_1$ (rad)}]
        \addplot[no markers] table [x=t, y=q1_safe, col sep=comma] {example_Data.csv};
        \addplot[no markers,red] table [x=t, y=q1_unsafe, col sep=comma] {example_Data.csv};
        \addplot[no markers,dotted,semithick] table [x=t, y=q1_ul, col sep=comma] {example_Data.csv};
        \addplot[no markers,dotted,semithick] table [x=t, y=q1_ll, col sep=comma] {example_Data.csv};
        \legend{with $u^*$, with $u_{\text{nom}}$}
      \end{axis}
      \begin{axis}[ylabel style={yshift=-2.5mm},label style={font=\footnotesize}, tick label style={
        font=\footnotesize}, yshift=-1.75cm,height=3.3cm,width=5.3cm,ylabel = {$\theta_2$ (rad)},xlabel = {time (s)}]
        \addplot[no markers] table [x=t, y=q2_safe, col sep=comma] {example_Data.csv};   
        \addplot[no markers,red] table [x=t, y=q2_unsafe, col sep=comma] {example_Data.csv};
        \addplot[no markers,dotted,semithick,forget plot] table [x=t, y=q2_ul, col sep=comma] {example_Data.csv};
        \addplot[no markers,dotted,semithick,forget plot] table [x=t, y=q2_ll, col sep=comma] {example_Data.csv}; 
      \end{axis}
    \end{tikzpicture}
    \caption{Time evolution of the link angles. The dotted lines are
      the safety constraints as seen from the safe trajectory.}
    \label{fig:thetas}
  \end{subfigure}
  \hspace{.01mm} 
  \begin{subfigure}{0.25\linewidth}
    \centering
    \begin{tikzpicture}
      \begin{axis}[ylabel style={yshift=-3mm},label style={font=\footnotesize}, tick label style={
        font=\footnotesize} ,legend columns=-1,legend pos = north west, legend style={draw=none,font=\tiny},width = 5.3cm,height = 5.2cm, xmin=-3.5, xmax = 3.5, ymin = -4.2, ymax = 3.5,ylabel = {$\theta_2$ (rad)},xlabel = {$\theta_1$ (rad)}]
        \addplot[no markers,semithick] table [x=q1_safe, y=q2_safe, col sep=comma] {example_Data.csv};
        \addplot[no markers,red,semithick] table [x=q1_unsafe, y=q2_unsafe, col sep=comma] {example_Data.csv}; 
        \draw [dotted, thick] (axis cs: -1.57,-1.57) node{}
        -- (axis cs: -1.57,1.57) node{}
        -- (axis cs: 0,3.14) node{}
        -- (axis cs: 1.57,1.57) node{}
        -- (axis cs: 1.57,-1.57) node{}
        -- (axis cs: 0,-3.14) node{}
        -- cycle;
        \legend{with $u^*$, with $u_{\text{nom}}$}
      \end{axis}    
    \end{tikzpicture}
    \caption{Trajectories on the position plane, showing the
      invariance enforced by the controller $u^*$ instead
      of~$u_{\text{nom}}$.}
    \label{fig:portrait}
  \end{subfigure}
  \hspace{.01mm} 
  \begin{subfigure}{0.25\linewidth}
    \centering
    \begin{tikzpicture}
      \begin{axis}[ylabel style={yshift=-2.2mm}, label style={font=\footnotesize}, tick label style={
        font=\footnotesize}, height=3.5cm, width=5.2 cm, xmajorticks=false,legend columns=-1,legend pos = north west, legend style={draw=none,font=\scriptsize},ymax=70,ylabel style={align=center},ylabel = input \\ magnitude \\ (N$\cdot$m)]
        \addplot[no markers] table [x=t_norm, y=u_norm_g1, col sep=comma] {example_Data.csv};
        \addplot[no markers,red] table [x=t_norm, y=u_norm_g.1, col sep=comma] {example_Data.csv};
        \legend{$\gamma = 1$, $\gamma = 0.1$}
        \end{axis}
        \begin{axis}[ylabel style={yshift=-2.2mm}, label style={font=\footnotesize}, tick label style={
          font=\footnotesize}, yshift=-1.75cm,height=3.3cm,width=5.2cm,ylabel style={align=center}, ylabel= velocity \\ magnitude \\ (rad/s),xlabel = {time (s)}]
        \addplot[no markers] table [x=t_norm, y=v_norm_g1, col sep=comma] {example_Data.csv};
        \addplot[no markers,red] table [x=t_norm, y=v_norm_g.1, col sep=comma] {example_Data.csv};
        \end{axis}
    \end{tikzpicture}
    \caption{Time evolution of the average velocity and
      average control effort on both links for different values of
      $\gamma$.}
    \label{fig:norms}
  \end{subfigure}
  \caption{Simulation results for safe control of a 2-link robot
    arm.}\label{fig:time-evols}
  \vspace*{-2.5ex}
\end{figure*}
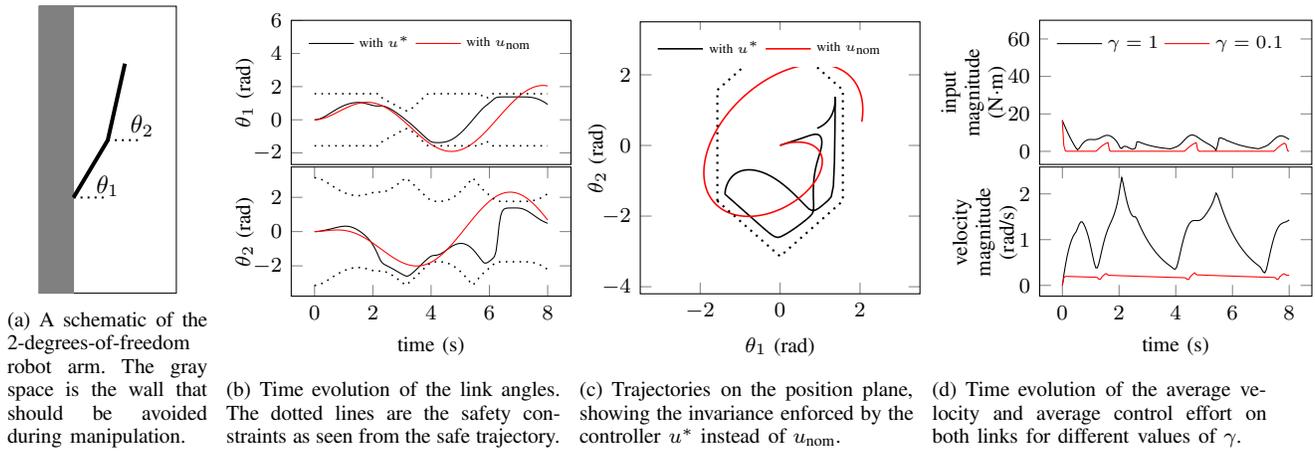

Figure~\ref{fig:thetas} shows the time evolution of the joint angles
under the nominal controller and under the safety-filtered
controller~$u^*$, along with the constraint evolution as viewed from
the position of the safe trajectory. Figure~\ref{fig:portrait} shows a
phase portrait of the evolution of the angles under the nominal and
the safe controllers. Note how the safe controller renders $\Cc$
invariant, as guaranteed by Corollary~\ref{cor:CIwFullControl}, while
the nominal controller does not.

We also note that the system~\eqref{eq:ex1Dynamics} satisfies
Assumption~\ref{as:stop} with
$f_2^1(x) = M^{-1}C \begin{bmatrix}\dot \theta_1 & \dot
  \theta_2 \end{bmatrix}^\top$,
$f_2^1(x)= M^{-1}\begin{bmatrix} 0 & c_{25}
  \cos(\theta_1+\theta_2)\end{bmatrix}^\top$ and $G_2(x_1) =
M^{-1}$. Thus, by Theorem~\ref{thm:limitedU}, any
$\Uc = \{u \in \real^2 \;|\; \|u\| \leq d\}$ such that $d > k_Gk_1$
suffices for invariance, with a sufficiently small $\gamma$ and a
suitably chosen $\epsilon$. This is reflected in
Figure~\ref{fig:norms}, which shows the effect of the design parameter
$\gamma$ on the control and velocity magnitudes.  As shown there,
lower values of $\gamma$ allow for safety with lower control
magnitudes, at the expense of reducing the velocity of the
execution. Finally, we note that extending this example to higher
degrees of freedom
adds complexity to the polytopic representation of the constraints
but, once available, the computation of the control design remains the
same.


\section{Conclusions}
Given a second-order system and positional safety specifications
described by linear boundaries, we have identified conditions that
allow the explicit construction of a verifiably safe set in the full
state space. We have also designed an associated QP controller that
ensures this set is safe. We have shown that the identified conditions
are always satisfied by fully actuated systems and, in the case of
Euler-Lagrange systems, we have shown how the controller design can
incorporate velocity and input constraints.  We believe the approach
presented here will be helpful in the design of safety-critical
controllers for second-order systems with other forms of
underactuation, something we plan to study in the future.  Other
extensions include incorporating time-varying safety considerations,
robustifying the approach to handle system uncertainties, and applying
a similar safe-set construction to specifications that are not
necessarily linear.


\bibliography{../bib/alias,../bib/Main-add,../bib/JC}
\bibliographystyle{unsrt}
\end{document}